\pdfoutput=1
\RequirePackage{ifpdf}
\ifpdf 
\documentclass[pdftex]{sigma}
\else
\documentclass{sigma}
\fi

\newtheorem{teo}{Theorem}
\newtheorem{prop}[teo]{Proposition}

\theoremstyle{definition}
\newtheorem{rmk}[teo]{Remark }

\begin{document}
	
\allowdisplaybreaks

\newcommand{\arXivNumber}{1606.06474}

\renewcommand{\PaperNumber}{081}

\FirstPageHeading
	
\ShortArticleName{Born--Jordan and Weyl Quantizations of the 2D Anisotropic Harmonic Oscillator}
	
\ArticleName{Born--Jordan and Weyl Quantizations\\ of the 2D Anisotropic Harmonic Oscillator}
	
\Author{Giovanni RASTELLI}
	
\AuthorNameForHeading{G.~Rastelli}
	
\Address{Dipartimento di Matematica, Universit\`a di Torino, Torino, via Carlo Alberto 10, Italy}
\Email{\href{mailto:giovanni.rastelli@unito.it}{giovanni.rastelli@unito.it}}

\ArticleDates{Received July 15, 2016, in f\/inal form August 15, 2016; Published online August 17, 2016}	
	
\Abstract{We apply the Born--Jordan and Weyl quantization formulas for polynomials in canonical coordinates to the constants of motion of some examples of the superintegrable 2D anisotropic harmonic oscillator. Our aim is to study the behaviour of the algebra of the constants of motion after the dif\/ferent quantization procedures. In the examples considered, we have that the Weyl formula always preserves the original superintegrable structure of the system, while the Born--Jordan formula, when producing dif\/ferent operators than the Weyl's one, does not.}
	
	
\Classification{81S05; 81R12; 70H06}

\section{Introduction}

Max Born and Pascual Jordan proposed a f\/irst model of quantization in~\cite{BJ}, restricted to one-dimensional systems, a model generalized to $n$-dimensional systems, together with Werner Heisenberg, in~\cite{BHJ}. Hermann Weyl introduced in~\cite{W,W1} a general quantization scheme based on the Fourier transform formula. The Born--Jordan and Weyl quantizations produce in general dif\/ferent operators from the same classical function of momenta and coordinates. The two quantization methods, however, coincide on natural Hamiltonians in $n$-dimensional Euclidean spaces. This property and many others of the two quantizations are studied in the recent book~\cite{dG}, that is our main source for the following discussion.

We consider here the Born--Jordan and Weyl quantizations applied to the algebra of constants of motion of the 2D anisotropic harmonic oscillator. This system admits the maximal number of three functionally independent constants of motion (and it is said to be superintegrable) whenever the ratio of the parameters is a rational number. Otherwise, the independent constants of motion are just two.

Our aim here is to check the behaviour of the Born--Jordan and Weyl quantizations of the anisotropic harmonic oscillator with respect to the integrability and superintegrability of the resulting quantum system for some particular choice of the ratio of the parameters. This aspect of the two quantization procedures is not considered in~\cite{dG}.

We use here the same expression for the classical 2D anisotropic harmonic oscillator and its independent constants of motion that is employed in \cite{CDRraz}. The nD anisotropic harmonic oscillator is there considered as an ``extended system'', a particular structure of some natural Hamiltonians that allows the existence of polynomial constants of motion of higher degree, described in~\cite{CDRraz, CDRTTW}. The interest in using such a construction here comes from other our studies in progress on the quantization of extended systems (see for example~\cite{CDRgen}).

We consider brief\/ly also the factorization in annihilation-creation operators for the 2D aniso\-tro\-pic harmonic oscillator as given in~\cite{JH}, to check how the corresponding classical integrals are quantized by applying again Born--Jordan and Weyl procedures. The examples are confronted with those arising from the extension procedure.

We can apply here the simpler formulas for Born--Jordan and Weyl quantizations for monomials in coordinates and momenta discussed in~\cite{dG}.

The quantization procedures of classical quantities are not exhausted by Born--Jordan's and Weyl's approaches. Several techniques have been developed since the beginning of the quantum era, many of them specif\/ic for the particular system considered.

Indeed, there is no unique way to assign quantum operators to classical quantities in a~mea\-ning\-ful way. Hermiticity of the operator is, usually, a~necessary requirement and it is obtained by some symmetrization procedure, however not uniquely determined, see for example~\cite{PW}.

The problem of preserving the algebra of the constants of motion of a Hamiltonian system after quantization is object of many recent studies, see for example~\cite{Herranz, KN, DV,MPW} and references therein, for solutions in f\/lat and non-f\/lat manifolds.

For superintegrability and quantization in classical and quantum systems see also~\cite{MPW}, in particular for the def\/inition of quantum superintegrable systems, and~\cite{KN}.

\section{Born--Jordan and Weyl quantizations of monomials}

In \cite{BJ,BHJ, dG} the Born--Jordan quantization of monomials in coordinates $(x_i,p_i)$ is determined by the following general rules,
\begin{gather*}
[\hat x_i,\hat p_j]=i\hbar \delta_{ij},\qquad [\hat x_i, \hat x_j]=0,\qquad [\hat p_i, \hat p_j]=0,
\end{gather*}
where $\delta_{ij}=1$ for $i=j$ and zero otherwise, for any quantization of the coordinates $x_i\rightarrow \hat x_i$ and of the momenta $p_i\rightarrow \hat p_i$, and by
\begin{gather}\label{BJ}
x_i^rp_i^s\rightarrow \frac 1{s+1}\sum_{k=0}^s \hat p_i^{s-k}\hat x_i^r\hat p_i^k,
\end{gather}
for the monomials with same indices. When the indices are dif\/ferent, the operators commute by the general quantization rules of above and their quantization is therefore straightforward.

For the Weyl quantization \cite{dG, W}, the general rules are the same of above and, for the monomials, we have instead
\begin{gather}\label{W}
x_i^rp_i^s\rightarrow \frac 1{2^s}\sum_{k=0}^s \binom{s}{k} \hat p_i^{s-k}\hat x_i^r\hat p_i^k.
\end{gather}

The standard realization of the operators $\hat x_i$ and $\hat p_i$, at least for Cartesian coordinates, are
\begin{gather*}
\hat x_i \phi= x_i \phi, \qquad \hat p_i \phi =-i\hbar \frac {\partial}{\partial x_i} \phi,
\end{gather*}
for any function $\phi(x_j)$. We employ in the following this standard quantization of the canonical coordinates.

It is easy to check that for $r=s=1$ the Born--Jordan and Weyl quantizations of monomials coincide, but dif\/fer for $r,s \geq 2$.

Many properties of the Born--Jordan quantization, generalized to any function of coordinates and momenta, and of the Weyl quantization are considered into details in~\cite{dG} and the dif\/ferent characteristics are discussed.

We focus here our analysis on the ef\/fect of the two quantizations on the f\/irst integrals of a~particular superintegrable system. We check in some examples if the quantized f\/irst integrals commute with the Hamiltonian operator, that is, if the algebraic structure of the constants of motion is preserved by the dif\/ferent formulas of quantization, an issue not considered in~\cite{dG}.

\section{The superintegrable 2D anisotropic harmonic oscillator}

In order to express the f\/irst integrals of the system, we can write the Hamiltonian of the superintegrable 2D anisotropic harmonic oscillator in the form of an extended Hamiltonian \cite{CDRraz} as follows
\begin{gather}\label{Hu}
H_{m,n}=\frac 12\left( p_u^2+\left(\frac mn\right)^2p_x^2\right)+\omega^2\left(\frac mn \right)^2\left(x^2+u^2\right),
\end{gather}
where $(x,u)$ are coordinates in the Euclidean plane, $\omega\in \mathbb R$ and $m,n \in \mathbb N\setminus \{0\}$.

Two independent f\/irst integrals of $H_{m,n}$ are $H_{m,n}$ itself and
\begin{gather*}
L=\frac 12 p_x^2+\omega^2x^2,
\end{gather*}
that is associated with the separability of the Hamilton--Jacobi equation of $H_{m,n}$ in coordina\-tes~$(u,x)$.

If we put
\begin{gather}\label{G}
G_n=\sum_{k=0}^{\left[\frac{n-1}{2}\right]}\binom{n}{2k+1}\big({-}2\omega^2\big)^k x^{2k+1}p_x^{n-2k-1},
\end{gather}
where $[a]$ denotes the integer part of $a$, and $X_L$ is the Hamiltonian vector f\/ield of the function~$L$, then, adapting to our case the more general theorem proved in~\cite{CDRraz}, we have

\begin{prop}\label{comp}
For any couple of positive integers $(m,n)$, the function $K_{m,n}$ is a first integral of $H_{m,n}$, where
\begin{gather*}
K_{m,n}=P_{m,n}G_n+D_{m,n}X_{L}(G_n),
\end{gather*}
with
\begin{gather*}
P_{m,n}=\sum_{k=0}^{[m/2]}\binom{m}{2k} \left(-\frac mn u \right)^{2k}p_u^{m-2k}\big({-}2\omega^2\big)^k,\\
D_{m,n}=\frac 1{n}\sum_{k=0}^{[(m-1)/2]}\binom{m}{2k+1}\, \left(-\frac mn u \right)^{2k+1}p_u^{m-2k-1}\big({-}2\omega^2\big)^k, \qquad m>1,
\end{gather*}
and $D_{1,n}=-\frac m{n^2}u$.
\end{prop}

We can introduce the usual Cartesian coordinates $(x,y)$ by leaving $x$, $p_x$ unchanged and putting
\begin{gather}\label{uy}
u=\frac nm y, \qquad p_u=\frac mn p_y,
\end{gather}
so that
\begin{gather*}
H_{m,n}=\left(\frac mn \right)^2\left(\frac 12\left( p_x^2+p_y^2 \right)+\omega^2\left(x^2+\left(\frac nm \right)^2y^2 \right) \right).
\end{gather*}
In the following, we consider the $H_{m,n}$ of above by dropping the negligible overall factor $\left(\frac mn \right)^2$.

The classical and quantum superintegrability of the $n$D anisotropic harmonic oscillator has been studied in~\cite{JH}. The quantization of the classical system is there obtained by intro\-du\-cing creation and annihilation operators. This technique is widely in use today (see for example~\cite{KN} and~\cite{Ba2}, where the anisotropic harmonic oscillator is generalized to 2D constant-curvature manifolds obtaining new classical and quantum superintegrable systems) and can have some application towards the quantization of extended systems.

The Jauch--Hill Hamiltonian of the anisotropic harmonic oscillator is
\begin{gather*}
H_{JH}=\frac 12\left(\frac{p_1^2}{M_1}+\frac{p_2^2}{M_2}+M_1\omega_1^2q_1^2+M_2\omega_2^2q_2^2 \right),
\end{gather*}
and coincide with $H_{m,n}$ if we put
\begin{gather*}
x=\sqrt{M_1}q_1, \qquad p_1=\sqrt{M_1}p_x,\qquad y=\sqrt{M_2}q_2, \qquad p_2=\sqrt{M_2}p_y,
\end{gather*}
with
\begin{gather*}
m\omega_2=n\omega_1, \qquad \omega_1=\sqrt{2}\omega.
\end{gather*}
The classical f\/irst integrals given in \cite{JH} become
\begin{gather}\label{F}
F_1(m,n)=\frac 12 \big(b_1^nb_2^{*m}+b_1^{*n}b_2^m \big),\qquad F_2(m,n)=-\frac i2 \big(b_1^nb_2^{*m}-b_1^{*n}b_2^m \big),
\end{gather}
where
\begin{gather*}
b_{1}=\frac 1{\sqrt{2\omega_1}}\left( p_x-i\omega_1 x \right),\qquad b_{1}^*=\frac 1{\sqrt{2\omega_1}}\left( p_x+i\omega_1 x \right),
\end{gather*}
and similarly $b_2$, $b_2^*$ in function of $y$, $p_y$, $\omega_2$. The corresponding quantum operators are obtained simply by substituting $p_x$, $p_y$ with $\hat p_x$, $\hat p_y$ in the expressions of above.

The f\/irst integrals obtained with the two methods of above are dif\/ferent, having for example, up to constant factors,
\begin{gather*}
K_{1,1}=xp_y-yp_x, \qquad F_1(1,1)=p_xp_y+\omega^2xy.
\end{gather*}

\section[The two quantizations of the constants of motion of the oscillator]{The two quantizations of the constants of motion\\ of the oscillator}

The Born--Jordan and Weyl quantizations of both $H_{m,n}$ and $L$ clearly coincide, being in both cases
\begin{gather*}
H_{m,n} \rightarrow  \hat H_{m,n}= -\frac {\hbar^2}2\left(\frac {\partial^2}{\partial x^2}+\frac {\partial^2}{\partial y^2} \right)+\omega^2\left(x^2+\left(\frac nm \right)^2y^2 \right),\\\
L \rightarrow \hat L=-\frac {\hbar^2}2 \frac {\partial^2}{\partial x^2}+\omega^2 x^2.
\end{gather*}
The operators are clearly independent and commuting, so that the integrability of the system is preserved by both the quantizations.

The Born--Jordan and Weyl quantizations of $K_{1,1}$, $K_{2,1}$ and $K_{3,1}$ coincide and commute with the corresponding Hamiltonian operators $\hat H_{m,n}$. Things become dif\/ferent for $(m,n)=(4,1)$. Indeed

\begin{prop} The first integral $K_{4,1}$ of $H_{4,1}$ is
\begin{gather}\label{Kmn}
K_{4,1}=256\left(xp_y^4-yp_xp_y^3-\frac 34\omega^2 xy^2p_y^2+\frac{\omega^2}8y^3p_xp_y+\frac {\omega^4}{64}xy^4\right).
\end{gather}
	
By applying to $K_{4,1}$ the Weyl quantization formula \eqref{W}, we have the Weyl operator
\begin{gather}
\hat K_{4,1}^{\rm W}=256\left(\hbar^4 \left( x \frac{\partial^4}{\partial y^4}- y \frac{\partial^4}{\partial x \partial y^3} \right)-\frac {3\hbar^4}2 \frac{\partial^3}{\partial x \partial y^2}+\frac {\hbar^2 \omega^2}8\left( 6xy^2 \frac{\partial^2}{\partial y^2} -y^3 \frac{\partial^2}{\partial x\partial y} \right)\right. \nonumber\\
 \left.\hphantom{\hat K_{4,1}^{\rm W}=}{}+ \hbar^2\omega^2y \left( \frac 32 x \frac{\partial}{\partial y} -\frac 3{16}y\frac{\partial}{\partial x} \right)+\frac{\omega^4}{64}xy^4+\frac {3\hbar^2 \omega^2}8 x\right),\label{A}
\end{gather}
and, by applying to $K_{4,1}$ the formula \eqref{BJ}, we get the Born--Jordan operator
\begin{gather}\label{B}
\hat K_{4,1}^{\rm BJ}=\hat K_{4,1}^{\rm W}+32\hbar^2\omega^2x.
\end{gather}
	
The computation of the commutators gives
\begin{gather}\label{comm}
\big[\hat H_{4,1},\hat K^{\rm W}_{4,1}\big]=0, \qquad \big[\hat H_{4,1},\hat K^{\rm BJ}_{4,1}\big]=-32\hbar^4\omega^2\frac{\partial}{\partial x}.
\end{gather}
\end{prop}

\begin{proof}We have from (\ref{G}) $G_1=x$, therefore $X_LG_1=p_x$. The application of Proposition~\ref{comp} to~(\ref{Hu})	with $(m,n)=(4,1)$ followed by the transformation of coordinates (\ref{uy}) gives
\begin{gather*}
P_{4,1}=256p_y^4-192\omega^2y^2p_y^2+4\omega^2y,\qquad D_{4,1}=-256yp_y^3+32\omega^2y^3p_y,
\end{gather*}
and we get (\ref{Kmn}). We observe now that, both for Born--Jordan and Weyl quantizations, we have
\begin{gather*}
\hat K_{4,1}= \hat P_{4,1}\hat x+\hat D_{4,1} \hat p_x,
\end{gather*}
where the order of the operators is immaterial, because $P_{4,1}$ and $D_{4,1}$ depend uniquely on~$(p_y,y)$. Therefore, we need to apply the two quantizations to $P_{4,1}$ and $D_{4,1}$ only, since the quantizations coincide on $x$ and $p_x$. By applying (\ref{BJ}) and (\ref{W}) to the monomials $y^2p_y^2$, $yp_y^3$ and $y^3p_y$, we have
\begin{gather*}
\hat P_{4,1}=256\hat p_y^4-192\omega^2\hat Q_1+4\omega^2\hat y,\qquad \hat D_{4,1}=-256\hat Q_2+32\omega^2\hat Q_3,
\end{gather*}
where
\begin{gather*}
\hat Q_2=i\frac {\hbar^3}2 \left( 2y\frac{\partial ^3}{\partial y^3}+3 \frac{\partial ^2}{\partial y^2} \right),\qquad
\hat Q_3=-i\frac \hbar 2 y^2 \left( 2 y \frac {\partial}{\partial y} +3\right),
\end{gather*}
coincide for both quantizations, and, for the Weyl case,
\begin{gather*}
\hat Q_1=-\frac {\hbar^2}2 \left(2y^2\frac{\partial ^2}{\partial y^2}+4y \frac {\partial}{\partial y} +1 \right),
\end{gather*}
	while, for the Born--Jordan case,
\begin{gather*}
\hat Q_1=-\frac {\hbar^2}3 \left(3y^2\frac{\partial ^2}{\partial y^2}+6y \frac {\partial}{\partial y} +2 \right).
\end{gather*}
We obtain in this way (\ref{A}) and (\ref{B}).
	
We can f\/inally compute the commutators of $\hat H_{4,1}$ with $\hat K_{4,1}^{\rm W}$ and $\hat K_{4,1}^{\rm BJ}$ with the help of the formula
	\begin{gather*}
	\big[\hat H_{4,1},\hat K_{4,1}\big]=\hat p_x\big({-}i\hbar \hat P_{4,1}+\big[\hat H_{4,1}, \hat D_{4,1}\big]\big) +\hat x \big( 2\omega^2 i \hbar \hat D_{4,1}+\big[\hat H_{4,1},\hat P_{4,1}\big]\big),
	\end{gather*}
	after making the suitable substitutions, obtaining the~(\ref{comm}).
\end{proof}	

We computed the quantizations for several other values of $(m,n)$, for example $(5,1)$, $(6,1)$, $(1,4)$, $(3,4)$, such that $\hat K_{m,n}^{\rm BJ}\neq \hat K_{m,n}^{\rm W} $, obtaining always commutation with $\hat H_{m,n}$ for $\hat K^{\rm W}_{m,n}$ and no commutation with $\hat H_{m,n}$ for $\hat K^{\rm BJ}_{m,n}$. In these last cases, it can be observed that both $\omega$ and $\hbar$ always appear as factors in the commutator, meaning that for $\omega=0$ and for $\hbar\rightarrow 0$ the operators commute.

Analogous results are obtained from the quantizations of the f\/irst integrals $F_1(m,n)$ or $F_2(m,n)$ given in (\ref{F}) for several values of $(m,n)$: the Weyl formula produces symmetry ope\-rators of the Hamiltonian, the Born--Jordan one, when giving dif\/ferent operators, does not.

Actually, one can conjecture that the Weyl quantizations of $K_{m,n}$ and $F_i(m,n)$ always commute with the Hamiltonian operator $\hat H_{m,n}$.

It can be observed, as one of the referees of this article pointed out, that ``for the Jauch--Hill approach the quantization using creation and annihilation operators shows very easily that the quantum extensions of $F_1$ and $F_2$ are still integrals. It
is less obvious to prove that this quantization is nothing but Weyl's one''.

\begin{rmk} The failure of the Born--Jordan quantization formula in reproducing the algebra of constants of motion at the quantum level is, actually, restricted to the particular set of generators of the algebra that we choose. We do not know in general if another choice of independent f\/irst integrals can lead to dif\/ferent results.
\end{rmk}

\begin{rmk} The failure in reproducing the algebra of the constants of motion after quantization appears also in the case of natural Hamiltonian systems on curved manifolds. In these cases, quantum corrections of the Hamiltonian operator are necessary in order to preserve the integrable or superintegrable algebraic structure \cite{Herranz, DV}.
\end{rmk}	

\section{Conclusions}

From the examples computed, it appears that the Born--Jordan quantization formula fails in preserving the high-degree constants of motion of the 2D anisotropic harmonic oscillator, and therefore its superintegrability, to the quantum level, dif\/ferently from the Weyl formula. A~study of this problem in full generality is then desirable.

\subsection*{Acknowledgements}
I am grateful to the referees of this article for their comments and suggestions.

\pdfbookmark[1]{References}{ref}
\LastPageEnding

\end{document}